\documentclass[letterpaper, 10 pt, conference]{ieeeconf}

\IEEEoverridecommandlockouts

\usepackage{graphicx}      
\usepackage{amsmath}
\usepackage{amssymb}
\usepackage{bbm}
\usepackage{dsfont}
\usepackage{algorithm}
\usepackage{algpseudocode}
\usepackage{xcolor}
\usepackage{fancyhdr}

\usepackage[utf8]{inputenc}
\usepackage[english]{babel}
\usepackage{lettrine}
\usepackage{geometry}
\newcommand{\subparagraph}{}
\usepackage[compact]{titlesec}
\titlespacing{\section}{0pt}{1ex}{1ex}
\titlespacing{\subsection}{0pt}{1ex}{1ex}
\usepackage[font=small,skip=2pt]{caption}
\newtheorem{lemma}{Lemma}
\newtheorem{remark}{Remark}
\newtheorem{assumption}{Assumption}

\newcommand{\ie}{\textit{i}.\textit{e}.}

\fancypagestyle{titlepage}{
	\fancyhf{} 
	\lhead{Submitted to ECC2025} 
}

\title{\LARGE \bf
Policy Gradient-based Model Free Optimal LQG Control with a Probabilistic Risk Constraint}
\author{Arunava Naha and Subhrakanti Dey
	\thanks{*This work is supported by The Swedish Research Council under grants 2017-04053.}
	\thanks{Arunava Naha and Subhrakanti Dey are with the Department of Electrical Engineering, Uppsala University, 75103 Uppsala, Sweden.
		{\tt\small e-mail: arunava.naha@angstrom.uu.se and subhrakanti.dey@angstrom.uu.se}}%
}

\begin{document}
\maketitle
\thispagestyle{titlepage}

\pagestyle{empty}

\begin{abstract}
In this paper, we investigate a model-free optimal control design that minimizes an infinite horizon average expected quadratic cost of states and control actions subject to a probabilistic risk or chance constraint using input-output data. In particular, we consider linear time-invariant systems and design an optimal controller within the class of linear state feedback control. Three different policy gradient (PG) based algorithms, natural policy gradient (NPG), Gauss-Newton policy gradient (GNPG), and deep deterministic policy gradient (DDPG), are developed and compared with the optimal risk-neutral linear-quadratic regulator (LQR), chance constrained LQR, and a scenario-based model predictive control (MPC) technique via numerical simulations. The convergence properties and the accuracy of all the algorithms are compared numerically. We also establish analytical convergence properties of the NPG and GNPG algorithms under the known model scenario, while the proof of convergence for the unknown model scenario is part of our ongoing work.

\end{abstract}

\section{Introduction} \label{sec:intro}
The linear quadratic regulator (LQR) problem has been extensively studied in the literature, and the optimal controller is known to be  a linear function of states \cite{bertsekasDynamicProgrammingOptimal}. However, the LQR formulation is risk neutral, \ie, it does not consider the risk or chance of occurrence of undesirable events. Such risky or undesirable events may occur due to the long tail of the process noise or disturbance. For a lot of practical control problems, it is desirable to avoid such risky events. For example, sometimes it is important for an unmanned aerial vehicle (UAV) to avoid flying over a certain area to hide from adversaries. Therefore, it is crucial to design a controller that minimizes the risk of such events along with minimizing the average expected control cost \cite{tsiamisRiskConstrainedLinearQuadraticRegulators2020}. Consider the wind turbine control problem, where wind speed introduces uncertainty, and the control aim is to optimize power output while mitigating structural damage risk \cite{schildbachScenarioApproachStochastic2014}. Similarly, in climate-controlled buildings, the objective is to minimize energy usage while ensuring occupant comfort and mitigating the risk of temperature exceeding set thresholds \cite{flemingStochasticMPCAdditive2019}. As studied in \cite{flemingStochasticMPCAdditive2019}, controllers designed with hard constraints are pessimistic compared to the ones designed with softer probabilistic constraints. In other words, the designed controller will lower the control cost if we constrain the probability of risky or undesirable events instead of imposing hard constraints. For example, in the case of the wind turbine control problem, the system will produce more power if we constrain the probability of the stress level on the blades increasing a prespecified limit instead of imposing a hard constraint on the stress level. 
\subsection{Related Work} \label{subsec:related_work}
In the model predictive control (MPC) literature, a probabilistic risk is generally modeled as a chance constraint. A popular approach is to draw samples or scenarios from the distribution of the disturbance and convert the probabilistic chance constraint into an algebraic one \cite{flemingStochasticMPCAdditive2019,schildbachScenarioApproachStochastic2014}. The probabilistic chance constraint is also handled by replacing it with an expected value using the Hamiltonian Monte Carlo (HMC) method or employing Chebyshev's inequality. In \cite{schildbach2015linear}, the chance constrained LQR problem in transformed into a convex optimization problem with linear matrix inequality (LMI) constraints and solved using semi-difinite programming (SDP). In a different approach, the chance constraints are formulated as the probability that the state and input values remain within certain sets, also called tubes \cite{arcariStochasticMPCRobustness2023}. Such tube-based chance constraint control is also studied under the model-free scenario in \cite{kerzDataDrivenTubeBasedStochastic2023}. Other than the MPC-based approaches, the occurrence of risky or undesirable events is also limited by constraining the average variance over an infinite time horizon of a quadratic function of states in \cite{zhaoGlobalConvergencePolicy2022,zhaoPrimaldualLearningModelfree2021}. The optimal controller under such risk formulation is proved to be an affine function of the states  \cite{zhaoInfinitehorizonRiskconstrainedLinear2021,tsiamisRiskConstrainedLinearQuadraticRegulators2020}.  

Reinforcement learning (RL) based techniques have performed remarkably for optimal decision-making problems, where the underlying system is partially or entirely unknown \cite{bertsekasReinforcementLearningOptimal2019, busoniuReinforcementLearningControl2018,lopezEfficientOffPolicyQLearning2023}. Policy gradient (PG) based actor-critic (AC) methods, a class of RL algorithms, are suitable for stochastic optimal control problems with continuous state and action spaces \cite{suttonReinforcementLearningSecond2018}. PG-based algorithms are also applied for the standard LQR problem, and their performance in terms of closed-loop stability and convergence are studied in the literature \cite{fazelGlobalConvergencePolicy2018a,huTheoreticalFoundationPolicy2023,yangProvablyGlobalConvergence2019}. It has been shown that although the optimization landscape is not convex with respect to the linear control gain in these problems, global convergence can be guaranteed when the PG algorithm is initialized with a stabilizing controller. The closed-loop stability and convergence analysis of PG-based algorithms for the LQR problem with additional constraints is challenging and has only been studied for a few specific cases. For example, in \cite{zhangPolicyOptimizationMathcal2021a}, the constraint is $H_\infty$ robustness constraint. The global convergence of a PG algorithm is studied for the risk-constrained LQR in \cite{zhaoGlobalConvergencePolicy2023a}, where the risk is modeled as an average variance of a quadratic function of states over an infinite time horizon. On the other hand, in \cite{hanReinforcementLearningControl2021}, the safety constraint is modeled as the expected value of a continuous non-negative function of the states being within a specified threshold, and the optimal controller is derived using an AC algorithm for a Markov decision process (MDP). In \cite{naha2023reinforcement}, a deep deterministic policy gradient (DDPG) based AC method is studied for the probabilistic risk-constrained LQR problem, where the risk is modeled as the probability that a quadratic function of the states crosses a user-defined limit.  However, analyzing the performance of PG-based AC algorithms for probabilistic risk-constrained LQR problems in general remains an open research area.

\subsection{Our Approach and Contributions} \label{subsec:contributions}
We have studied the performance of three PG-based AC algorithms, natural policy gradient (NPG) \cite{rajeswaranGeneralizationSimplicityContinuous2017,kakadeNaturalPolicyGradient2001}, Gauss-Newton policy gradient (GNPG) \cite{fazelGlobalConvergencePolicy2018a}, and deep deterministic policy gradient (DDPG) \cite{lillicrapContinuousControlDeep2016}, for the probabilistic risk- or chance-constraint LQR problem under the unknown model scenario. We have investigated the optimal policy within the class of linear state feedback controls, where a Lagrangian-based primal-dual formulation is used to handle the constraint. Furthermore, we have proved that the optimization problem under study enjoys coercivity and gradient dominance properties, and the NPG and GNPG algorithms converge to the global optimum under the known model assumption. The coercivity and L-smoothness properties also ensure that intermediate policies will maintain closed-loop stability while training, provided we start from an arbitrary stabilizing controller. The theoretical study on the convergence of the DDPG algorithm and the convergence properties of the NPG and GNPG algorithms under the unknown model scenario (where sample-based estimates of relevant quantities are used in the PG update) is left for our future publication.

We evaluate PG-based AC policies against risk-neutral LQR, a model-based chance constrained LQR and scenario-based MPC through simulations. As anticipated, PG-based AC policies effectively reduce risky events, albeit with a slight increase in quadratic cost compared to standard LQR. On the other hand, the model-based chance constrained LQR and the MPC performed comparably to the proposed model-free PG-based methods. However, both are model based approaches, and the effectiveness of MPC depends heavily on the time horizon length chosen, increasing computational complexity. Unlike MPC, model-free PG-based methods do not require real-time optimization at each step, relying solely on a feedforward actor-network post-training, significantly reducing computational overhead compared to MPC-based methods. 

{\color{black} 
Furthermore, we have modeled the risk as the probability that a function of the one-step ahead future state crosses a user-defined limit, and the risk is constrained by keeping the average expected violation probability over an infinite time horizon within a user-defined limit. We have used the indicator function in the reward structure to replace the probability when the system model is unknown. As discussed in the Related work section, the probability of the risky events is also limited indirectly by modeling the risk as an average variance of a quadratic function of states over an infinite time horizon \cite{zhaoGlobalConvergencePolicy2023a} or an expected value of a continuous non-negative function of the states \cite{hanReinforcementLearningControl2021}. Such formulations yield a closed-form analytical expression of the constraint function. On the other hand, our constraint model directly puts a bound on the probability of risky events. Furthermore, in general, such a probabilistic constraint does not give a closed-form analytical expression of the constraint function, which makes the convergence proofs more challenging and requires substantial new analysis.  

We can summarize our main contributions as follows. 
\begin{enumerate}
	\item To the best of our knowledge, the convergence property of NPG and GNPG is studied for the first time for the probabilistic risk or chance-constrained LQR problem, even for the known model scenario. While their convergence has been studied for the risk-neutral LQR formulation in existing literature, extending the convergence study to the chance constraint LQR requires significant new analysis.
	\item We have studied a Lagrangian-based primal-dual formulation to handle the constraint and proved that there is no duality gap.
	\item We have performed numerical simulation-based comparative analyses of NPG, GNPG, DDPG, MPC, and standard LQR for the probabilistic risk-constrained LQR problem. 
\end{enumerate}}

\subsection{Organization} \label{subsec:organization}
The rest of the paper is organized as follows. In Section \ref{sec:problem}, we present the problem formulation and the reward structure for the probabilistic risk or chance constrained LQR problem. In Section \ref{sec:ac_method}, we present the NPG, GNPG, and DDPG algorithms, while the convergence properties of NPG and GNPG are studied in Section~\ref{sec:main_results}. In Section \ref{sec:results}, we present the numerical results and compare the performance of the PG-based AC algorithms with the risk-neutral LQR and the scenario-based MPC. Finally, we conclude the paper in Section \ref{sec:conclusion}.

\subsection{Notations}
\label{subsec:notations}
Some special notations are given in Table~\ref{tab:notations}. 
\begin{table}[h!]
	\begin{center}
		\caption{Notations}
		\label{tab:notations}
		\begin{tabular}{l|l} 
			\hline \hline
			Symbol & Description \\
			\hline
			${\rm I\!R}^{n}$ & The set of $n\times 1$ real vectors \\
			${\rm I\!R}^{m\times n}$ & The set of $m\times n$ real matrices \\
			${X}^T$ & Transpose of matrix or vector ${X}$ \\
			$\mathcal{N}(\mu,{\bf \Sigma})$ & Gaussian distribution with mean $\mu$ and variance $\bf \Sigma$  \\
			$\bf \Sigma \ge 0$ or $>0$  & $\bf \Sigma$ is positive semi-definite or definite matrix, respectively \\
			$[\cdot]_{ij}$ & $i$-th row and $j$-th column element of a matrix \\
			$||\cdot||$ & Frobenius norm of a matrix of Euclidean norm of a vector \\
			$\text{tr}(\cdot)$ & Trace of a matrix \\
			$\text{E}\left[\cdot\right]$ and $\text{P}\left\{\cdot \right\}$ & Expectation operator and Probability measure respectively \\
            $\left\{\hat \cdot \right\}$ & Estimated or approximated value \\
			$\mathbbm{1}_{\left\{ condition \right\}}$ & Indicator function, 1 if condition is true, 0 otherwise \\
            $\sigma(\cdot)$ and $\rho(\cdot)$ & Singular value and eigenvalue of a matrix, respectively \\
			\hline \hline 
		\end{tabular}
	\end{center} \vspace{-4mm}
\end{table}

\section{Problem Formulation} \label{sec:problem}
We consider the following linear time-invariant (LTI) system:
\begin{equation}
	{\bf{x}}_{k+1}={\bf A}{\bf{x}}_{k}+{\bf B}{\bf{u}}_{k}+{\bf{w}}_{k}.
	\label{eqn:state_eqn}
\end{equation}
\vspace{-2mm}Here ${\bf{x}}_{k}\in {\rm I\!R}^{n}$ and ${\bf{u}}_{k}\in {\rm I\!R}^{p}$ are the state and input vectors at the $k$-th time instant respectively, whereas ${\bf{w}}_{k} \in {\rm I\!R}^{n}$ is an independent and identically distributed  (iid) process noise with distribution $f_w({\bf w})$. ${\bf{A}}\in {\rm I\!R}^{n\times n}$, ${\bf{B}}\in {\rm I\!R}^{n\times p}$. 

We assume that all the states are measured and the system $({\bf A},{\bf B})$ is stabilizable. In the standard LQR problem, the following cost function is minimized. 
\begin{align}
J=\lim_{T\to\infty} \frac{1}{T} \sum_{k=1}^T \text{E}\left[ {\bf x}^T_k{\bf Q}{\bf x}_k+{\bf u}^T_k{\bf R}{\bf u}_k\right],
	\label{eqn:cost_fun_J}
\end{align}
where ${\bf{Q}}\in {\rm I\!R}^{n\times n}$ and ${\bf{R}}\in {\rm I\!R}^{p\times p}$ are positive definite weight matrices. We also assume that $({\bf A},{{\bf Q}^{1/2}})$ is detectable. If the noise is zero mean and the second-order moment of the noise is bounded, then the optimum input appears as a fixed gain linear control signal \cite{bertsekasDynamicProgrammingOptimal}, see (\ref{eqn:opt_u}).
\begin{align}
	{\bf u}^*_k={\bf K}{ {\bf x}}_{k} \text{, and }  
	{\bf K}=-\left( {\bf B}^T{\bf S}{\bf B}+{\bf R}\right)^{-1}{\bf B}^T{\bf S}{\bf A} \text{,} \label{eqn:opt_u}
\end{align}
where $\bf S$ is the solution to the following algebraic Riccati equation, ${\bf S}={\bf A}^T{\bf S}{\bf A}+{\bf Q}-{\bf A}^T{\bf S}{\bf B}\left({\bf B}^T{\bf S}{\bf B}+{\bf R}\right)^{-1}{\bf B}^T{\bf S}{\bf A}.$
However, as discussed before, the cost formulation (\ref{eqn:cost_fun_J}) does not take into account the less frequent but risky events. Therefore,  we use an additional constraint on the probability of risky or undesirable events, and the optimization problem takes the following form. 
\begin{equation}
    \begin{aligned}
        & \underset{{\bf u} \in \mathcal{U}}{\text{minimize}}
        & & J \\
        & \text{subject to}
        & & J_c \leq \delta , \label{eqn:opt_prob}
    \end{aligned}
\end{equation}
where $J$ is same as given in (\ref{eqn:cost_fun_J}), and $J_c$ is given as follows.
\begin{equation}
    \begin{aligned}
        J_c = \lim_{T\to\infty} \frac{1}{T} \sum_{k=1}^T E\left[P\left\{f_c\left({\bf x}_{k+1} \right)\ge \epsilon \mid \Psi_{k} \right\}  \right]. \label{eqn:J_c}
    \end{aligned}
\end{equation}
Here, $\epsilon >0$ and $\delta >0$ are user selected parameters. $\Psi_{k} \triangleq \left \{{\bf x}_{l}, {\bf u}_{l} \mid k \ge l \ge 0 \right \}$ denotes the set of all information up to the instant $k$. 
\begin{remark}
	The risky or undesirable event is modeled as the function $f_c(\cdot)$ of the state at the next time step crossing a threshold $\epsilon$, and we are interested in limiting the probability of these events. Since such probability itself is a function of the random information set $\Psi_{k}$, we have taken the expectation with respect to this set in the above formulation. Furthermore, we are interested in keeping the long-term average probability bounded over an infinite time horizon.
\end{remark}
\subsection{Reward Structure} \label{subsec:reward}
The constrained optimization of (\ref{eqn:opt_prob}) can be converted into an unconstrained stochastic control problem using the Lagrangian multiplier $\lambda$ as follows, \vspace{-3.5mm}
\begin{equation} 
    \begin{aligned}
     \underset{{\bf u} \in \mathcal{U}}{\text{min }}
	 \mathcal{L} = J + \lambda \left(J_c - \delta \right) = \lim_{T\to\infty} \frac{1}{T} \sum_{k=1}^T \text{E}\left[g\left( {\bf x}_k,{\bf u}_k\right) \right],
	\label{eqn:JL}
    \end{aligned}  
\end{equation}
where the per stage cost $g(\cdot)$ takes the following form, 
\begin{align}
	&g\left( {\bf x}_k,{\bf u}_k\right) = f\left( {\bf x}_k,{\bf u}_k\right) +\lambda \left( h_p\left( {\bf x}_k,{\bf u}_k\right) - \delta \right) \text{, where} \label{eqn:gk} \\
	&f\left( {\bf x}_k,{\bf u}_k\right) = {\bf x}^T_k{\bf Q}{\bf x}_k+{\bf u}^T_k{\bf R}{\bf u}_k,  \label{eqn:fk} \\
	& h_p\left( {\bf x}_k,{\bf u}_k\right) = P\left\{f_c\left({\bf x}_{k+1} \right)\ge \epsilon \mid \Psi_{k} \right\}.  \label{eqn:hpk} 
\end{align}
Note that the per-stage cost function may generally contain an intractable probabilistic constraint. Therefore, for the RL-based algorithms, where we have access to the future states ($\bf x_{k+1}$) in the form of stored data, the reward is defined as 
\begin{align}
    &r_k = -f\left( {\bf x}_k,{\bf u}_k\right) -\lambda \left( h_r\left( {\bf x}_{k+1}\right) - \delta \right), \text{ where}\label{eqn:rk} \\
    &h_r\left( {\bf x}_{k+1}\right) = \mathds{1}_{\left\{f_c\left({\bf x}_{k+1} \right)\ge \epsilon \right\}}. \label{eqn:hrk}
\end{align}
Here, $\mathds{1}_{\left\{\cdot \right\}}$ is the indicator function, which takes the value $1$ if the condition inside the bracket is true, and $0$, otherwise. In the following section, we will briefly introduce the PG-based AC algorithms used in our study.
\section{PG-based AC Algorithms under study} \label{sec:ac_method}
In this section, we will present NPG, GNPG, and DDPG algorithms, the three PG-based AC algorithms used in our study. For the NPG and GNPG algorithms, we assume the policy to be stochastic but stationary, denoted by ${\bf u}_k \sim \pi_{\theta}(\cdot|\bf x_k)$, where $\theta$ is the policy parameter. The policy is deterministic for the DDPG algorithm, denoted by ${\bf u}_k = \mu_{\theta}(\bf x_k)$, where $\theta$ is the policy parameter. We will use the general notation $p(\bf x_k)$ to denote the stochastic or deterministic policy. In general, the policy parameter is updated using the gradient of the expected return, \ie, $\mathcal{R}$, see (\ref{eqn:R}), with respect to the policy parameter. 
\begin{equation}
    \mathcal{R} = \lim_{T\rightarrow \infty} \frac{1}{T} \text{E}\left[\sum_{k=1}^T r_k \right].
    \label{eqn:R}
\end{equation}
Furthermore, the value function (\ref{eqn:value_fn}), the Q function (\ref{eqn:Q_fn}) and the advantage function (\ref{eqn:Adv_fn}) under a policy $p(\cdot)$ are defined as follows. We have used ($\hat{\cdot}$) notation to indicate an estimated or approximated quantity. Note that even though the reward in (\ref{eqn:rk}) is a function of the future state ${\bf x}_{k+1}$, for the known model case, we can write the reward as a function of the current state ${\bf x}_k$ and the control input ${\bf u}_k$ using (\ref{eqn:state_eqn}).  \vspace{-3.0mm}
\begin{align}
    &V^{p}({\bf x}_k) = \lim_{T\rightarrow \infty}\sum_{i=k}^T\left\{\text{E}\left[ r_i \mid {\bf x}_k \right]- \mathcal{R}\right\}, \label{eqn:value_fn} \\
    &Q^{p}({\bf x}_k,{\bf u}_k) = \text{E}\left[ r_k + V^{p}({\bf x}_{k+1}) \mid {\bf x}_k,{\bf u}_k \right]. \label{eqn:Q_fn} \\
    &A^{p}({\bf x}_k,{\bf u}_k) = Q^{p}({\bf x}_k,{\bf u}_k) - V^{p}({\bf x}_k). \label{eqn:Adv_fn}
\end{align}
\subsection{Natural Policy Gradient (NPG) based AC algorithm} \label{subsec:npg} 
We have adopted the NPG-based AC algorithm from \cite{rajeswaranGeneralizationSimplicityContinuous2017}; see Algorithm~\ref{algo:NPG}. NPG methods utilize the Fisher information matrix, $F$, to obtain the steepest ascent direction as $F^{-1}G$, where $G$ is the gradient of the expected return, $\mathcal{R}$ with respect to the policy parameter $\theta$. In practice, $G$ is estimated using the policy gradient theorem \cite{rajeswaranGeneralizationSimplicityContinuous2017} from the data as follows, \vspace{-3.5mm}
\begin{equation}
    \hat G = \frac{1}{N}\sum_{k=1}^N \hat{A}({\bf x}_k,{\bf u}_k) \nabla_{\theta} \log \pi_{\theta}({\bf u}_k|{\bf x}_k).
    \label{eqn:G}
\end{equation}
Here, $\hat{A}(\cdot)$ is the estimated advantage function. Similarly, $F$ is estimated as follows \cite{rajeswaranGeneralizationSimplicityContinuous2017},
\begin{equation}
    \hat F = \frac{1}{N}\sum_{k=1}^N \nabla_{\theta} \log \pi_{\theta}({\bf u}_k|{\bf x}_k) \nabla_{\theta} \log \pi_{\theta}({\bf u}_k|{\bf x}_k)^T.
    \label{eqn:F}
\end{equation} 
The step size for the policy parameter update is evaluated in the same way as \cite{rajeswaranGeneralizationSimplicityContinuous2017}, ensuring the policy parameter update is not too large, see Algorithm~\ref{algo:NPG}. Furthermore, the advantage value is estimated using the method provided in \cite{schulmanHighDimensionalContinuousControl2018} as follows, \vspace{-2.5mm}
\begin{equation}
    \begin{aligned}
    \hat{A}({\bf x}_k,{\bf u}_k) = \sum_{l=0}^{T-1}(\gamma \eta)^l d_{k+l},
    \label{eqn:Adv} \text{ where} \\
    d_{k+l} = -{\hat V}_{\phi}({\bf x}_{k+l}) + r_{k+l} + \eta {\hat V}_{\phi}({\bf x}_{k+l+1}). 
    \end{aligned}
\end{equation}
Here, $0<\eta <1$ and $0<\gamma<1$. ${\hat V}_{\phi}(\cdot)$ denotes the value obtained from the critic network parameterized by $\phi$. The value function parameter $\phi$ is updated using the steepest descent direction as $\hat H^{-1} s$, where $s$ and $\hat H$ are evaluated as
\begin{align}
    & s = \nabla_{\phi} \left( \frac{1}{N}\sum_{k=1}^N\mid \mid {\hat V}_{\phi}({\bf x}_k) - {\hat V}_k \mid \mid^2\right), \text{ and} \label{eqn:s} \\
    & \hat H = \frac{1}{N}\sum_{k=1}^N \mathcal{J}_k \mathcal{J}_k^T, \text{ where } \mathcal{J}_k = \nabla_{\phi}{\hat V}_{\phi}({\bf x}_k).\label{eqn:H} 
\end{align}
\begin{equation}
    \text{The target value ${\hat V}_k$ is evaluated as }{\hat V}_k = \sum_{l=0}^{T-1} (\gamma )^l r_{k+l}.
    \label{eqn:Vk}
\end{equation}   
Similar to the policy parameter update, the step size for the critic parameter update is also obtained in such a way that the update is not too large, see Algorithm~\ref{algo:NPG}.        
\begin{algorithm}[h!]
    \caption{NPG-based AC Algorithm}
    \label{algo:NPG}
    \begin{algorithmic}[1]
        \State Initialize policy parameter $\theta_0$ and the value function parameter $\phi_0$.
        \State Set $\gamma$, $\eta$, $\alpha_a$, $\alpha_c$.
        \State Set Number of timestep data used for advantage and value evaluations $N$, and Number of trajectories $M$.
        \For{$i=1,2,...$}
            \For{$j=1$ to $M$}
                \State Generate a trajectory $\left\{ {\bf x}_k,{\bf u}_k,r_k \right\}_{k=1}^T$ using the policy $\pi_{\theta_i}(\cdot|\bf x_k)$.
                \State Compute $d_k$ using (\ref{eqn:Adv}) 
                \For {$k=1$ to $N$}
                    \State Compute $\hat{A}({\bf x}_k,{\bf u}_k)$ (\ref{eqn:Adv}) and $\nabla_{\theta} \log \pi_{\theta_i}({\bf u}_k|{\bf x}_k)$. 
                    \State Compute ${\hat V}_k$ (\ref{eqn:Vk}) and $\mathcal{J}_k$ (\ref{eqn:H}). 
                \EndFor
                \State Compute ${\hat G}_j$ and ${\hat F}_j$ using (\ref{eqn:G}) and (\ref{eqn:F}), respectively.
                \State Compute ${\hat H}_j$ and $s_j$ using (\ref{eqn:H}) and (\ref{eqn:s}), respectively. 
            \EndFor
            \State Compute ${\hat G} = \frac{1}{M}\sum_{j=1}^M {\hat G}_j$ and ${\hat F} = \frac{1}{M}\sum_{j=1}^M {\hat F}_j$.
            
            \State Update the policy parameter $\theta_{i+1} \leftarrow \theta_i + \sqrt \frac{\alpha_a}{{\hat G}^T{\hat F}^{-1} {\hat G}}  {\hat F}^{-1} {\hat G}$.
            \State Compute ${\hat H} = \frac{1}{M}\sum_{j=1}^M {\hat H}_j$ and $\hat s = \frac{1}{M}\sum_{j=1}^M {s}_j$.
            \State Update the value function parameter $\phi_{i+1} \leftarrow \phi_i + \sqrt \frac{\alpha_c}{{\hat s}^T{\hat H}^{-1} {\hat s}}  {\hat H}^{-1} {\hat s}$.
        \EndFor
    \end{algorithmic}
\end{algorithm}


\subsection{Gauss-Newton Policy Gradient (GNPG)} \label{subsec:gnpg}
GNPG algorithm is a variant of the NPG algorithm, where the estimated Fisher information matrix $\hat F$ is replaced by the estimated Gauss-Newton matrix $\hat H_a$, see (\ref{eqn:Ha}), otherwise all the other steps are the same as Algorithm~\ref{algo:NPG}, \cite{rajeswaranGeneralizationSimplicityContinuous2017,schulmanHighDimensionalContinuousControl2018}. \vspace{-3.0mm}
\begin{equation}
    \begin{aligned}
    {\hat H}_{a,j} = \frac{1}{N}\sum_{k=1}^N g_{a,j} g_{a,j}^T, \text{ where}   \\
    g_{a,j} = \hat{A}({\bf x}_k,{\bf u}_k) \nabla \log \pi_{\theta_i}({\bf u}_k|{\bf x}_k).
    \label{eqn:Ha}
    \end{aligned}
\end{equation}

\subsection{Deep Deterministic Policy Gradient (DDPG)} \label{subsec:ddpg}
The DDPG-based AC method  is based on the Algorithm~\ref{algo:DDPG} \cite{lillicrapContinuousControlDeep2016,naha2023reinforcement}. The actor and critic networks are parameterized by $\theta$ and $\phi$, respectively. The actor takes states as input and outputs control inputs, while the critic takes states and control inputs, providing a Q value for that state-action pair. In DDPG, there are two separate target networks, $Q^t$ and $\mu^t$, for the critic and actor, respectively. These target networks facilitate stable learning by offering consistent targets and are updated gradually to track the main networks, as described in Algorithm~\ref{algo:DDPG}.

\begin{algorithm}[h!]
    \caption{DDPG-based AC Algorithm}
    \label{algo:DDPG}
    \begin{algorithmic}[1]
        \State Set $\tau$, learning rates $\alpha_d$, initial and final variances of zero mean Gaussian noise for exploration ($\mathcal N_t$), \ie, $\Sigma_{D,0}$ and $\Sigma_{D,F}$.  
        \State Set Number of timestep data used for loss evaluation $N$, and Number of episodes $M$.   
        \State Initialize $\theta_0$ and $\phi_0$, and $\theta_0^t \leftarrow \theta_0$ and $\phi_0^t \leftarrow \phi_0$. 
        \State Initialise the replay buffer $\mathcal D$.  
        
        \For{episode = 1, M}
            \State Receive initial observation state ${\bf x}_1$.
            \For{k = 1, T}
                \State Select action ${\bf u}_k = \mu_\theta({\bf x}_k) + \mathcal N_k$ [$\mathcal N_k$ is zero mean Gaussian noise].
                \State Execute action ${\bf u}_k$ and observe reward $r_k$ and observe new state ${\bf x}_{k+1}$.
                \State Store transition $\left({\bf x}_k, {\bf u}_k, r_k, {\bf x}_{k+1} \right)$ in $\mathcal D$.

                \State Sample a random minibatch of $N$ transitions $\left({\bf x}_i, {\bf u}_i, r_i, {\bf x}_{i+1} \right)$ from $\mathcal D$.
                
                \State Set $y_i = r_i + \gamma {\hat Q}_{\phi^t}^t \left( {\bf x}_{i+1}, \mu_{\theta^t}^t \left( {\bf x}_{i+1} \right)   \right)$.
                
                \State Update critic by minimizing the loss: $L = \frac{1}{N} \sum_i \left( y_i - {\hat Q}_{\phi} \left( {\bf x}_i, {\bf u}_i \right) \right)^2.$                
                \State Update the actor policy using the sampled policy gradient as \newline $\nabla_{\theta} \mathcal{R} \approx \frac{1}{N} \sum_i \nabla_u {\hat Q}_{\theta} \left( {\bf x}, {\bf u} \right)\mid_{{\bf x} = {\bf x}_i, {\bf u} = \mu_{\theta}({\bf x}_i)} \nabla_{\theta} \mu_{\theta} \left({\bf x}  \right)\mid_{{\bf x} = {\bf x}_i}$
                \State Update the target networks as \newline 
                $\theta^{t} \leftarrow \tau \theta + (1-\tau) \theta^{t}\ \text{and}\  \phi^{t} \leftarrow \tau \phi + (1-\tau) \phi^{t}$
            \EndFor   
            \State Reduce the variance of Gaussian noise for exploration until it reaches its final value 
        \EndFor 
    \end{algorithmic}
\end{algorithm} 
In the following section, we  present some analytical results on the convergence properties of the NPG and GNPG algorithms.
\section{Analytical Results} \label{sec:main_results} 
We investigate the two fundamental properties for the theoretical analysis of PG-based AC methods for the optimal controller design problem given by (\ref{eqn:opt_prob}). The first property is the convergence of the AC algorithm to a local or global optimum policy and the corresponding convergence rate, while the second property concerns the closed-loop stability of the system during the training process. We investigate the convergence properties and stability aspects of the NPG and GNPG algorithms under the known model scenario and for the linear state feedback control. The study for more general cases, such as unknown model scenarios (where only sampled based estimates of the relevant quantities are available) and the convergence analysis of the DDPG algorithm for the probabilistic risk constrained control problem, is part of ongoing work.  

We assume the policy to have the following form, \vspace{-2.0mm}
\begin{equation}
    \pi_{K}(\cdot |{\bf x}) = \mathcal{N}\left(-{\bf K}{\bf x}, \Sigma_{\sigma} \right), \label{eqn:policy}
\end{equation} 
where $\mathcal{N}(\cdot,\cdot)$ denotes the Gaussian distribution. Furthermore, $\bf K$ is a trainable parameter, and $\Sigma_{\sigma}$ is a fixed covariance matrix. Additionally, for theoretical analysis, we assume zero-mean Gaussian process noise, i.e., ${\bf w}_k \sim \mathcal{N}\left({\bf 0}, \Sigma_{w} \right)$. We anticipate that analogous theoretical outcomes can be derived for Gaussian mixture process noise, although that is a part of our ongoing research.  Finally, the control input at $k$-th time instant can be written as, \vspace{-5.0mm}
\begin{equation}
    {\bf u}_k = -{\bf K}{\bf x}_k + \sigma_k \text{, } \sigma_k \sim \mathcal{N}\left(0, \Sigma_{\sigma} \right). \label{eqn:u}
\end{equation}
We also define the set of all stabilizing linear state feedback controllers as $\mathcal{K} \triangleq \left\{ {\bf K} \mid \rho\left({\bf A}-{\bf B}{\bf K}\right) < 1 \right\}$, where $\rho(\cdot)$ denotes the spectral radius. Under the policy (\ref{eqn:u}), the closed-loop system dynamics can be written as,
\begin{align}
    {\bf x}_{k+1} &= \left({\bf A}-{\bf B}{\bf K}\right){\bf x}_k + { \bf \bar w}_k, \text{where}  \label{eqn:cl_sys} \\
    {\bf \bar w}_k &={\bf w}_k + {\bf B}\sigma_k. \label{eqn:w_bar}
\end{align}
Here, ${\bf \bar w}_k \sim \mathcal{N}\left(0, \Sigma_{\bar w} \right)$, where $\Sigma_{\bar w} = \Sigma_{w} + {\bf B}\Sigma_{\sigma}{\bf B}^T$, which can be derived directly using (\ref{eqn:state_eqn}) and (\ref{eqn:u}). Additionally, we assume the following function for the constraint,
\begin{align}
    f_c\left({\bf x}_{k+1} \right) = {\bf q}^T{\bf x}_{k+1}, \label{eqn:fq}
\end{align} 
where ${\bf q} \in {\rm I\!R}^{n}$ is a user defined vector.

Before discussing our theoretical results, we rewrite the Lagrangian function from (\ref{eqn:JL}) using the control input given by (\ref{eqn:u}) as follows. \vspace{-3.0mm}
\begin{align}
    \mathcal{L}(\bf K, \lambda) &= J(\bf K) + \lambda \left(J_c(\bf K) - \delta \right) \text{, where}  \label{eqn:LK} \\
    J(\bf K) &= \text{tr}\left( \left(\bf Q + {\bf K}^T{\bf R}{\bf K}  \right)\Sigma_K + {\bf R}\Sigma_{\sigma}  \right) \label{eqn:JK1} \\
    &= \text{tr}\left({\bf P}_K \Sigma_{\bar w} +  {\bf R}\Sigma_{\sigma} \right) \label{eqn:JK2} \text{, and} \\
    J_c(\bf K) &= \text{E}\left[Q\left(a({\bf x}_k, {\bf K})   \right)\right] \label{eqn:JcK} \text{, where} \\
    Q(a) &= \frac{1}{\sqrt{2\pi}}\int_{a}^{\infty} e^{-\frac{z^2}{2}}dz \label{eqn:Qa} \text{, and} \\
    a({\bf x}_k, {\bf K}) &= \frac{\epsilon - {\bf q}^T({\bf A} - {\bf B}{\bf K}){\bf x}_k}{\sqrt{{\bf q}^T\Sigma_{\bar w}{\bf q} } } \label{eqn:a}
\end{align}
If ${\bf K} \in \mathcal{K}$, then $\Sigma_K$ and ${\bf P}_K$ are the unique solutions to the Lyapunov equations given in (\ref{eqn:Sigma_K}) and (\ref{eqn:P_K}), respectively.
\begin{align}
    \Sigma_K &= \Sigma_{\bar w} + ({\bf A} - {\bf B}{\bf K})\Sigma_K({\bf A} - {\bf B}{\bf K})^T \label{eqn:Sigma_K} \text{, and} \\
    {\bf P}_K &= {\bf Q} + {\bf K}^T{\bf R}{\bf K} + ({\bf A} - {\bf B}{\bf K})^T{\bf P}_K({\bf A} - {\bf B}{\bf K}) \label{eqn:P_K}
\end{align}
\begin{remark}
    The derivations of (\ref{eqn:JK1}), (\ref{eqn:JK2}), (\ref{eqn:Sigma_K}) and (\ref{eqn:P_K}) are available in \cite{yangProvablyGlobalConvergence2019}. It is straightforward to derive (\ref{eqn:JcK}), (\ref{eqn:Qa}) and (\ref{eqn:a}) using (\ref{eqn:cl_sys})-(\ref{eqn:fq}) in (\ref{eqn:J_c}) and considering the states, $\left\{\bf x_k \right\}$ to be ergodic.
\end{remark}

Our first result is the following lemma, which states the coercivity property of the Lagrangian function $\mathcal{L}(\bf K, \lambda)$ given by (\ref{eqn:LK}).
\begin{lemma}
    For a fixed $\lambda > 0$, the Lagrangian function $\mathcal{L}(\bf K, \lambda)$ given by (\ref{eqn:LK}) is coercive on $\mathcal K$ in the sense that $\mathcal{L}(\bf K, \lambda) \rightarrow \infty$ as $\bf K \rightarrow \delta \mathcal{K}$, where $\delta \mathcal{K}$ denotes the boundary of $\mathcal K$. 
    \label{lemma:coercive}  
\end{lemma}
\begin{proof}
    The proof follows from the fact that the cost function $J(\cdot)$ is coercive on $\mathcal K$, see \cite{huTheoreticalFoundationPolicy2023}, and the constraint function $0 \le J_c(\cdot) \le 1$ is bounded.
\end{proof}
\begin{remark}
    The coercivity property of $\mathcal{L}(\bf K, \lambda)$ is crucial to ensure the stability of the closed-loop system during the training process. In other words, the coercive function $\mathcal{L}(\bf K, \lambda)$ serves as a barrier function over the stable policy set $\mathcal K$, and no additional measure is required to ensure the stability of the closed-loop system during the training process.
\end{remark}
To demonstrate the convergence of the NPG and GNPG algorithms to a local or global optimum, it is imperative to establish the gradient dominance property of the Lagrangian function $\mathcal{L}(\mathbf{K}, \lambda)$ (\ref{eqn:LK}). Furthermore, the following two lemmas are required to support this property.
\begin{lemma}
    For a given $\lambda > 0$, the Lagrangian function $\mathcal{L}(\bf K, \lambda)$ given by (\ref{eqn:LK}) is twice continuously differentiable over $\mathcal K$.
    \label{lemma:twice_diff}
\end{lemma}
\begin{proof}
    The cost function $J(\cdot)$ is twice continuously differentiable over $\mathcal K$, see \cite{huTheoreticalFoundationPolicy2023}. Since the exponential function is analytic, $Q(a)$ is also analytic in $a$. Finally, $a({\bf x}_k, {\bf K})$ is an affine function of $\bf K$, so we can say $J_c(\bf K)$ is an analytic function of ${\bf K} \in {\mathcal K}$, and hence $\mathcal{L}(\bf K, \lambda)$ is at least twice continuously differentiable with respect to $\bf K$ over $\mathcal K$.
\end{proof}
\begin{lemma}
    For a given $\lambda > 0$, the Lagrangian function $\mathcal{L}(\bf K, \lambda)$ given by (\ref{eqn:LK}) is $L$-smooth on $\mathcal K_\zeta$, where $L >0$ is a constant and depends on the problem parameters and $\zeta$. Here, $\mathcal K_\zeta \triangleq \left\{ {\bf K} \in \mathcal{K} \mid \mathcal{L}(\bf K, \lambda) \le \zeta \right\}$ is a compact subset.
    \label{lemma:Lsmooth}
\end{lemma}
\begin{proof}
    Using Lemma~\ref{lemma:coercive} and Lemma~\ref{lemma:twice_diff} in Theorem~1 from \cite{huTheoreticalFoundationPolicy2023}, we can directly state Lemma~\ref{lemma:Lsmooth}.
\end{proof}
\begin{remark}
    The $L$-smoothness of $\mathcal{L}(\bf K, \lambda)$ as stated in Lemma~\ref{lemma:Lsmooth} also means 
    \begin{align}
         \mid \mid \nabla_K \mathcal{L}({\bf K}, \lambda)  \mid \mid^2 \le L\text{ } \forall \bf K \in \mathcal 
         K_\zeta.
    \end{align}
$||\cdot||$ denotes the Frobenius norm for a matrix or the Euclidean norm for a vector.
\end{remark}
To establish a linear convergence rate for the NPG and GNPG algorithms, we need the $L$-smoothness and gradient dominance properties of the Lagrangian function $\mathcal{L}(\bf K, \lambda)$. The following lemma states the gradient dominance property of $\mathcal{L}(\bf K, \lambda)$.
\begin{lemma}
    (Gradient dominance) For a given $\lambda > 0$, the Lagrangian function $\mathcal{L}(\bf K, \lambda)$ given by (\ref{eqn:LK}) satisfies the following inequality,
    \begin{align}
        \mathcal{L}({\bf K}, \lambda) - \mathcal{L}({\bf K}^*, \lambda) \le \mu \mid \mid \nabla_K \mathcal{L}({\bf K}, \lambda)  \mid \mid^2 \text{, } \forall {\bf K} \in \mathcal K, \label{eqn:grad_dom}
    \end{align}
    where $\mu>0$ is a constant, and $\bf K^* \in \mathcal K$ is the optimal policy parameter for a fixed $\lambda$, \ie, ${\bf K}^* = \arg \min_{\bf K \in \mathcal{K}} \mathcal{L}({\bf K}, \lambda)$.
    \label{lemma:grad_dom}
\end{lemma}
\begin{proof}
    The proof of Lemma~\ref{lemma:grad_dom} is provided in Appendix~\ref{apdx:grad_dom}.
\end{proof}
Finally, we state the convergence rate result for the NPG algorithm as follows, which is a direct consequence of Lemma~\ref{lemma:Lsmooth} and Lemma~\ref{lemma:grad_dom}. The
convergence of the GNPG algorithm can be shown following similar steps and a detailed proof will be provided in our ongoing work. {\color{black}
\begin{lemma}
    (Convergence rate) For a given $\lambda > 0$, the NPG algorithm converge to a global optimal policy parameter ${\bf K}^*$ with a linear convergence rate, \ie,
    \begin{align}
        \mathcal{L}({\bf K}^{'} \lambda) - \mathcal{L}({\bf K}^*, \lambda) \le \beta (\mathcal{L}({\bf K}, \lambda) - \mathcal{L}({\bf K}^*, \lambda)) \label{eqn:conv_rate}
    \end{align}
    Here ${\bf K}^{'}$ is the NPG update from $\bf K$ in a single iteration, see (\ref{eqn:K_update}) \cite{yangProvablyGlobalConvergence2019}. The constant $0<\beta <1$ depends on the problem parameters and learning rate $\alpha > 0$. 
    \begin{equation}
        {\bf K}^{'} = {\bf K} - \alpha [{F}]^{-1} \nabla_K \mathcal{L}({\bf K}, \lambda) = {\bf K} - \alpha  \nabla_K \mathcal{L}({\bf K}, \lambda) \Sigma_K^{-1}. \label{eqn:K_update}
    \end{equation}
    Here ${F}$ is the Fisher information matrix. $[{F}]_{(i,j)(i^{'},j^{'})} = \text{E}[\nabla_{K_{ij}}\log(\pi_K({\bf u}|{\bf x}))\nabla_{K_{i^{'}j^{'}}}\log(\pi_K({\bf u}|{\bf x}))^T] $. 
    \label{lemma:conv_rate}
\end{lemma}
\begin{proof}
    The proof of Lemma~\ref{lemma:conv_rate} is provided in Appendix~\ref{apdx:convergence_rate}.
\end{proof}
\begin{remark}
    Note that the variant of NPG algorithm  as given in Algorithm~\ref{algo:NPG} is $\alpha = \sqrt \frac{\alpha_a}{{\nabla_K \mathcal{L}({\bf K}, \lambda)}^T{[F]}^{-1} {\nabla_K \mathcal{L}({\bf K}, \lambda)}}$, where $\alpha_a > 0$ is a user selected parameter \cite{rajeswaranGeneralizationSimplicityContinuous2017}. The minor difference in the stepsize expression from Algorithm~\ref{algo:NPG} is due to the fact that the unknown parameter $\theta$ in Algorithm~\ref{algo:NPG} is a vector, but in the proof, $\bf K$ is a matrix.   
\end{remark}}
\begin{remark}
    The convergence rate lemma (\ref{eqn:conv_rate}) means if we start from an arbitrary stabilizing controller ${\bf K}_0 \in \mathcal{K}$, and update the policy parameter ${\bf K}_i$ using the NPG algorithm, then ${\bf K}_i \rightarrow {\bf K}^*$ as $i \rightarrow \infty$ at an exponential rate.
\end{remark}

In summary, so far, we have proved that the NPG algorithm converges linearly to a global optimal policy parameter $\bf K^*$ that minimizes $\mathcal{L}(\bf K, \lambda)$ for a given $\lambda>0$ when the relevant quantities in the algorithms are computed assuming true model knowledge. {\color{black}
\subsection{Finding the optimal value of the Lagrange multiplier $\lambda$}
We follow a primal-dual approach to find an optimal value of the Lagrange multiplier $\lambda$, see Algorithm~\ref{algo:primal_dual}. The dual problem is defined as follows \vspace{-3.0mm}
\begin{align}
    \max_{\lambda \ge 0} D(\lambda) =  \max_{\lambda \ge 0} \min_{\bf K \in \mathcal K} \mathcal{L}(\bf K, \lambda) \label{eqn:dual_prob}.
\end{align}
\begin{algorithm}[h!]
    \caption{Primal-Dual Algorithm}
    \label{algo:primal_dual}
    \begin{algorithmic}[1]
        \State Initialize $\lambda_0$ and $\alpha_{\lambda,0}$.
        \For{$i=0,1,2,\ldots$}
            \State Solve the primal problem $\bf K_{i} = \operatorname*{arg\,min}_{\bf K \in \mathcal K} \mathcal{L}(\bf K, \lambda)$ using NPG or GNPG, see Algorithm~\ref{algo:NPG}.
            \State Evaluate $\nabla_{\lambda} \mathcal{L}(\bf K_i, \lambda) = J_c({\bf K}_i) - \delta$.
            \State Update $\lambda_{i+1} = \max\left(0, \lambda_i + \alpha_{\lambda,i} \nabla_{\lambda} \mathcal{L}(\bf K_i, \lambda) \right)$. 
        \EndFor
    \end{algorithmic}
\end{algorithm} 

In Algorithm~\ref{algo:primal_dual}, $\alpha_{\lambda,i} >0$, $\alpha_{\lambda,i} = \mathcal{O}(i^{-1/2})$, is the learning rate for the Lagrange multiplier $\lambda$.} To prove that the pair ($\bf K^*, \lambda^*$) is also the optimal solution to the primal constrained problem (\ref{eqn:opt_prob}), we need Assumption~\ref{assump:slater} and Lemma~\ref{lemma:duality}.

\begin{assumption}
    (Slater's condition) There exists a $\bf \bar K \in \mathcal K$ such that $J_c(\bf \bar K) < \delta$.
    \label{assump:slater}
\end{assumption}
\begin{lemma}
    (Strong duality) Under
    Assumption~\ref{assump:slater}, the optimal value of the primal problem (\ref{eqn:opt_prob}) is equal to the optimal value of the dual problem (\ref{eqn:dual_prob}), \ie, $J^* = D^*$.
    \label{lemma:duality}
\end{lemma}
\begin{proof}
    The proof of Lemma~\ref{lemma:duality} is provided in Appendix~\ref{apdx:duality}.
\end{proof} {\color{black}
\begin{remark} (\textbf{Convergence of Algorithm~\ref{algo:primal_dual}}) Based on Lemma~\ref{lemma:coercive} and Lemma~\ref{lemma:grad_dom}, we can say that the controller ${\bf K}_i$ from Algorithm~\ref{algo:primal_dual} will always produce a stabilizing controller. This implies that both $\nabla_{\lambda} \mathcal{L}(\bf K_i, \lambda)$ and $\lambda_i$ will be bounded by some positive constants. Furthermore, according to Theorem 4 in \cite{zhaoGlobalConvergencePolicy2023}, we can conclude that Algorithm~\ref{algo:primal_dual} will converge to an optimal policy at a sublinear rate, given that the step size $\alpha_{\lambda,i} = \mathcal{O}(i^{-1/2})$.
\end{remark}}
\begin{remark} {\color{black}
    Note that in this paper, we only studied the convergence properties of the NPG and GNPG algorithms under the known model scenario, which does not require any critic network. In other words, the value function is computed using the true model knowledge. On the other hand, the proof for the unknown model case is challenging since we approximate the value function and estimate the advantage and the required gradients from the input-output data.} For the model-free scenario, in general, it is first proved that the value function approximator will converge and provide a close approximation of the true value of the state after a sufficiently large number of iterations, say $i \ge I_T$ \cite{yangProvablyGlobalConvergence2019}. Then it can be shown that $|\mathcal{L}({\bf K}_{i+1},\lambda) - \mathcal{L}({\tilde {\bf K}}_{i+1},\lambda)|$ is small for $i \ge I_T$, where $\mathcal{L}({\bf K}_{i+1},\lambda)$ and $\mathcal{L}({\tilde {\bf K}}_{i+1},\lambda)$ are the Lagrangian functions evaluated using the true and approximated value functions, respectively. Combining this with the convergence rate Lemma~\ref{lemma:conv_rate}, we can show that the NPG and GNPG algorithms converge to a global optimal policy parameter $\bf K^*$ with a linear convergence rate for the unknown model case. However, detailed proofs for the unknown model case are part of our ongoing work. 
\end{remark}

\section{Numerical Results} \label{sec:results}
In this section, we compare the performance of the NPG, GNPG, and DDPG algorithms with the risk-neutral LQR, chance constrained LQR (CLQR) and the scenario-based MPC by numerical simulations. As our case study, we have considered an unmanned aerial vehicle (UAV) model \cite{zhaoGlobalConvergencePolicy2023a}, a fourth-order LTI system. The UAV model parameters and the parameter values used for the simulation study are provided in Appendix~\ref{apdx:params}. All three PG-based algorithms use the same network structures for actor and critic networks. The actor is just a linear function of the state, and the critic is a fully connected neural network with two hidden layers of size (10,50). We have used the $\tanh$ activation function for the hidden nodes, and the outer layer has no activation function. 

The chance constrained LQR and scenario-based MPC algorithm used for the comparison are given in Algorithm ~\ref{algo:SDP} \cite{schildbach2015linear} and Algorithm~\ref{algo:MPC} \cite{schildbachScenarioApproachStochastic2014}, respectively.
\begin{algorithm}[h!]
    \caption{Chance constrained LQR (CLQR)}
    \label{algo:SDP}
    \begin{algorithmic}
        \State Apply SDP to solve the following optimization problem and obtain the optimal controller as ${\bf K_{sdp}} ={\bf Y} {\bf X}^{-1}$.
        \begin{equation*}
            \begin{aligned}
                &\min_{{\bf{X}}, {\bf Y}, {\bf P}} \quad \text{Tr}({\bf QX}) + \text{Tr}({\bf P})\\
                &\text{s.t.} \quad \begin{bmatrix}
                    {\bf P} & ({\bf R}^{1/2} {\bf Y}) \\
                    ({\bf R}^{1/2} {\bf Y})^T & {\bf X}
                \end{bmatrix} \succeq 0,\\
                &\quad \begin{bmatrix}
                    {\bf X} - {\bf W} & {\bf A}{\bf X} + {\bf B}{\bf Y} \\
                    ({\bf A}{\bf X} + {\bf B}{\bf Y})^T & {\bf X}
                \end{bmatrix} \succeq 0,\\
                &\quad q^T{\bf X} q \leq \alpha \epsilon^2, \text{ where } \alpha = (nrm^{-1}(1-\delta))^{-2}.
            \end{aligned}
        \end{equation*}
        \State Here $nrm(\cdot)$ is the cumulative normal distribution function. ${\bf{X}} \in {\rm I\!R}^{n\times n} >0$, ${\bf{P}} \in {\rm I\!R}^{p\times p} \geq 0$, and ${\bf{Y}} \in {\rm I\!R}^{p\times n}$.
    \end{algorithmic}
\end{algorithm}



\begin{algorithm}[h!]
    \caption{Scenario-based chance constraint MPC}
    \label{algo:MPC}
    \begin{algorithmic}
        \State Perform the following steps at each time step $t$:
    
        \State Measure the current state ${\bf x}_t$.
    
        \State Generate $S$ noise samples ${\bf w}_t^{(1)}, \cdots, {\bf w}_t^{(S)} \sim f_w({\bf w})$.
    
        \State Solve the following optimization problem:
        \begin{equation*}
            \begin{aligned}
                &\min_{{\bf u}_{1|t},\cdots,{\bf u}_{T|t}} \sum_{s=1}^S \sum_{i=1}^T \left(f\left( {\bf x}_{i|t}^{(s)},{\bf u}_{i|t}\right) + \lambda \mathds{1}_{\left\{f_c\left({\bf x}_{i+1|t}^{(s)}\right) \ge \epsilon\right\}}   \right)\\
                &\text{s.t. (\ref{eqn:state_eqn}) is satisfied. }
            \end{aligned}
        \end{equation*}
    
        \State Apply the first control input ${\bf u}_{1|t}$ to the system. 
        \State [$f(\cdot|\cdot)$ is given in (\ref{eqn:fk}). ${\bf x}_{i|t}$ and ${\bf u}_{i|t}$ denote predictions and plans of the state and input variables made at time $t$, for $i$ steps into the future.]
    \end{algorithmic}
\end{algorithm}
First, we compare the training performance of the PG-based algorithms by plotting average return $\mathcal{R}$ with respect to the training timestep data from a trial run in Fig.~\ref{fig:return}. The average return is evaluated from the separate test data using the trained models after each $5\times 10^5$ timestep training data sample. When starting from the same stable controller, we observe that NPG and GNPG algorithms maintain closed-loop stability while training. On the other hand, DDPG does not necessarily maintain closed-loop stability while training under similar conditions. However, in some instances DDPG is seen to reach the optimal policy faster. 
\begin{figure}[h!]
    \centering
    \includegraphics[width=65mm]{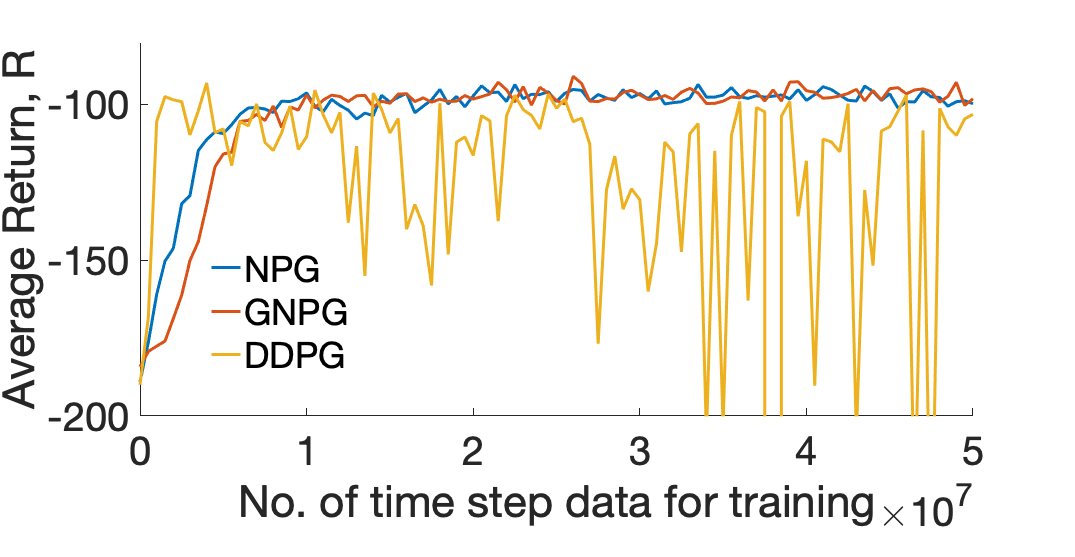}
    \caption{Average return $\mathcal{R}$. $\lambda = 100$, $J_c = 8.7\%$ (NPG).}
    \label{fig:return}
\end{figure}
Figure~\ref{fig:prob_and_cost} compares the constraint violation probability $J_c$ and the control cost $J$ for different $\lambda$ values for all five algorithms. All the PG-based algorithms are trained over $5\times 10^7$ timestep samples. Because of their stochastic nature, we have taken the average of the actor-network parameter of the best ten training iterations to generate the test results in Fig.~\ref{fig:prob_and_cost}.
\begin{figure}[h!]
    \centering
    \includegraphics[width=75mm]{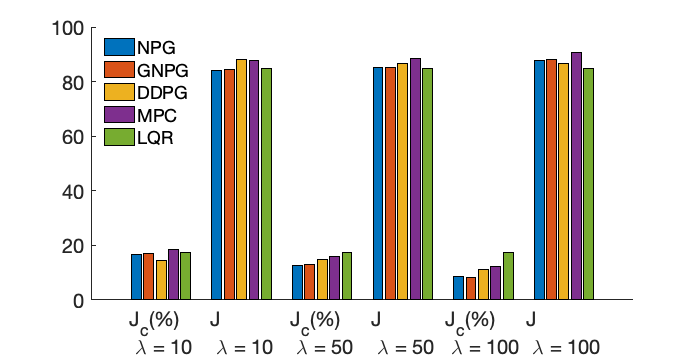}
    \caption{Constraint violation probability $J_c$ and control cost $J$ for different $\lambda$ values.}
    \label{fig:prob_and_cost}
\end{figure}


As expected, the constraint violation percentage gets reduced for the PG-based methods at the expense of a small increase in the quadratic cost when compared with the standard LQR. Moreover, the performance of the MPC method has a  similar trend to that of PG-based algorithms. However, it is crucial to note that MPC's performance is heavily dependent on the chosen parameters \( S = 20 \) and \( T = 5 \) in Algorithm 3. While increasing these parameters can enhance MPC's performance, it comes with the trade-off of increased computational complexity, which is of the order of $ST^2$ \cite{skaf2009nonlinear}. In addition, MPC necessitates solving an optimization problem at every time step, in contrast to the PG-based methods, which only require evaluating the feed-forward actor-network. This distinction renders PG-based methods significantly less computationally complex than MPC. It is also important to acknowledge that MPC is a model-based method, which further differentiates it from the PG-based techniques. We also observed that the results from the DDPG algorithm  do not show a clear trend as the other algorithms with the increase in $\lambda$. Because of this reason, we have excluded DDPG from the comparative plot in Fig.~\ref{fig:prob_vs_cost}, where the control cost $J$ for different $\lambda$ is plotted with respect to the constraint violation probability $J_c$. For CLQR, the $J_c$ values from NPG is used as thresholds, \ie, $\delta = J_c/100$. We observe that NPG and GNPG algorithms perform almost similarly, outperforming the MPC for the given \( S = 20 \) and \( T = 5 \). Additionally, the proposed PG-based methods performed very similar to the CLQR method. Note that, CLQR is a model-based approach.

\begin{figure}[h!]
    \centering
    \includegraphics[width=45mm]{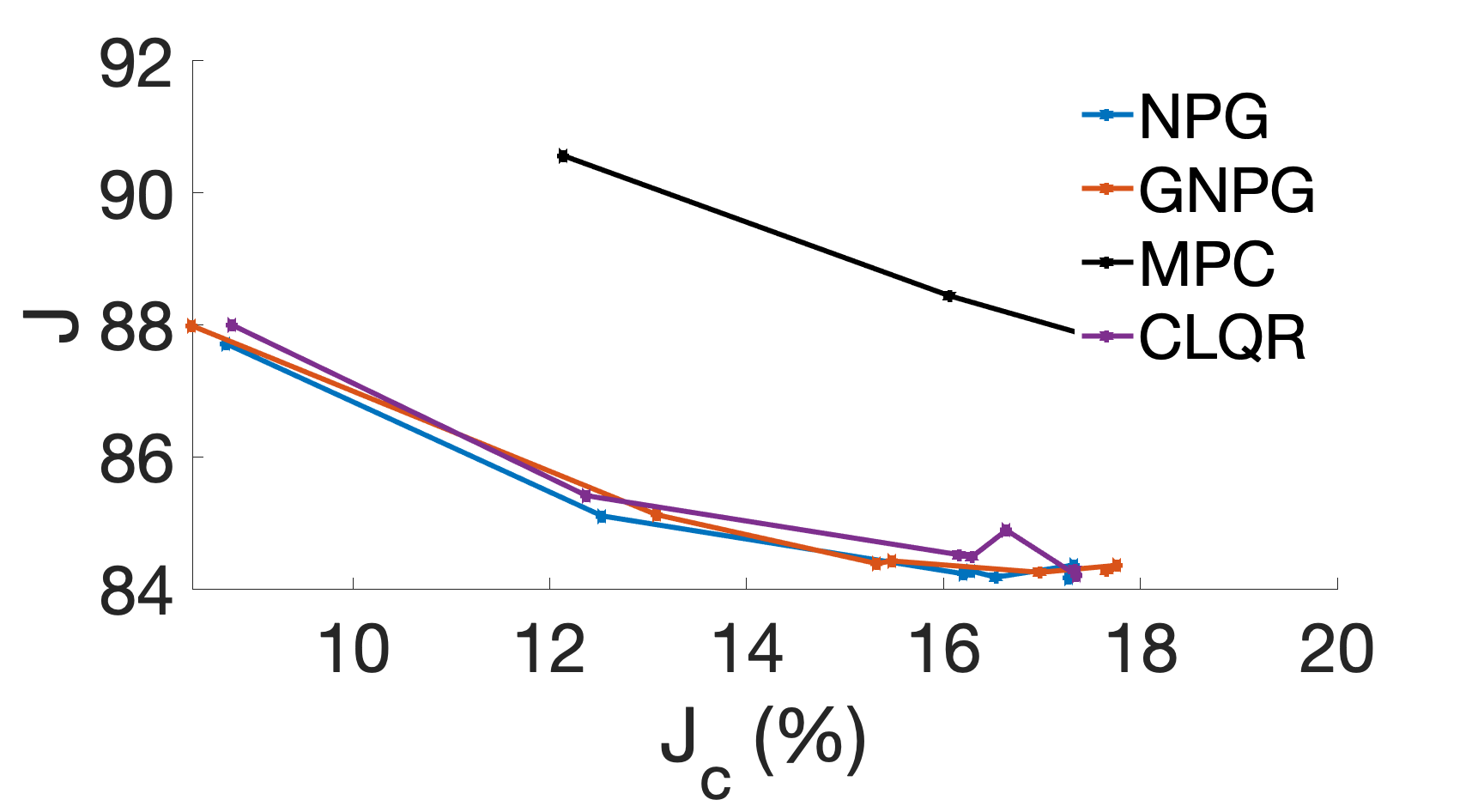}
    \caption{Control cost $J$ vs. Constraint violation probability $J_c$. $\lambda = [1,5,10,15,20,50,100]$.}
    \label{fig:prob_vs_cost}
\end{figure}

Figure~\ref{fig:policy_grad} and Fig.~\ref{fig:critic_loss} compare the norm of the policy gradient, \ie, $||\hat G||$, and critic loss with respect to the number of training iterations for the NPG and GNPG algorithms. The plot shows the mean and 95\% confidence interval of the quantity. We observe that NPG has a marginally better convergence rate compared to GNPG.
\begin{figure}[h!]
    \centering
    \includegraphics[width=50mm]{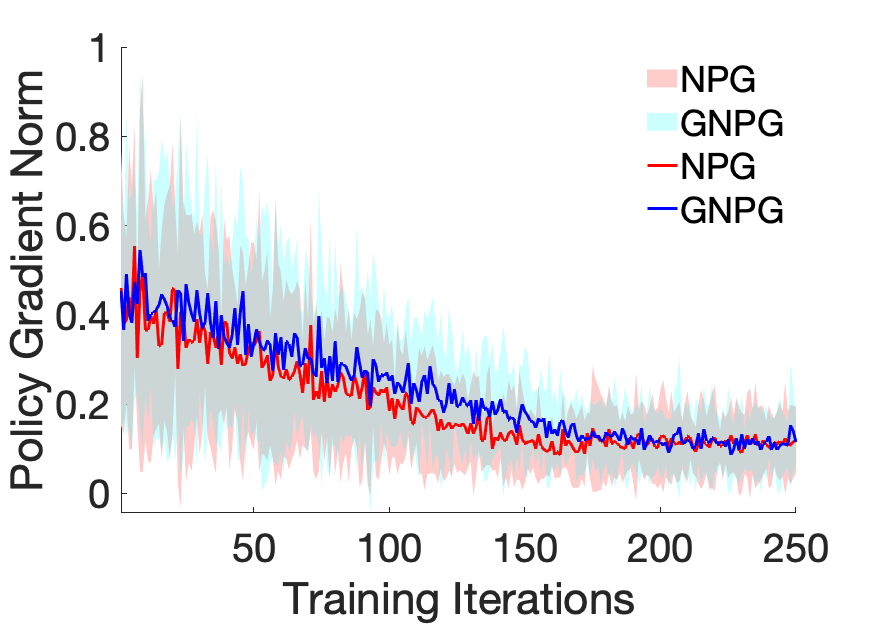}
    \caption{Norm of the policy gradient. $\lambda = 10$, $J_c = 16.5\%$ (NPG).}
    \label{fig:policy_grad}
\end{figure}
\begin{figure}[h!]
    \centering
    \includegraphics[width=50mm]{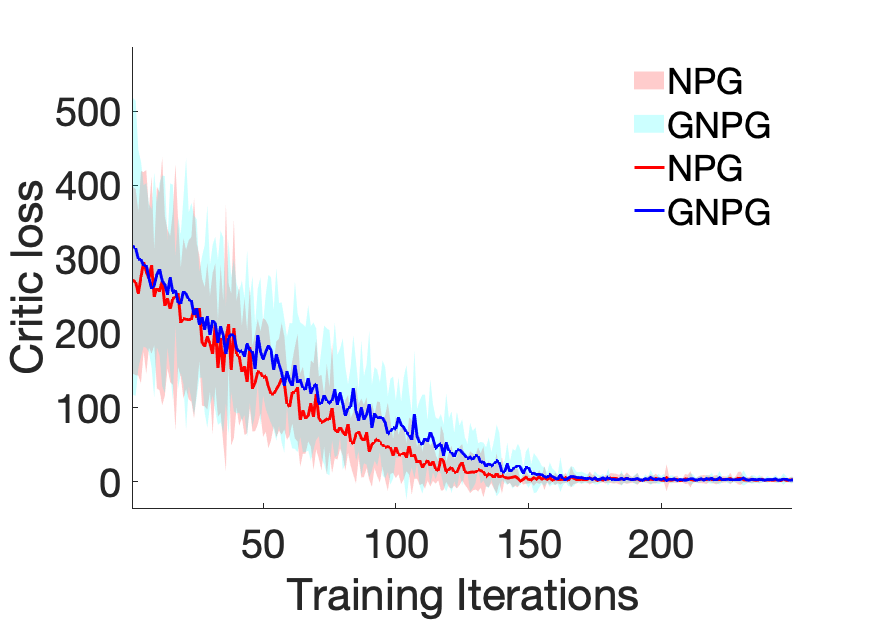}
    \caption{Critic loss. $\lambda = 10$, $J_c = 16.5\%$ (NPG).}
    \label{fig:critic_loss}
\end{figure}
{\color{black}
In Fig.~\ref{fig:sdp_vs_pd}, we compared the primal-dual algorithm, Algorithm~\ref{algo:primal_dual}, for NPG-based AC with the CLQR, Algorithm~\ref{algo:SDP}, for the same threshold value $\delta$, and observed slightly lower control cost $J$ for the CLQR.} 
\begin{figure}[h!]
    \centering
    \includegraphics[width=50mm]{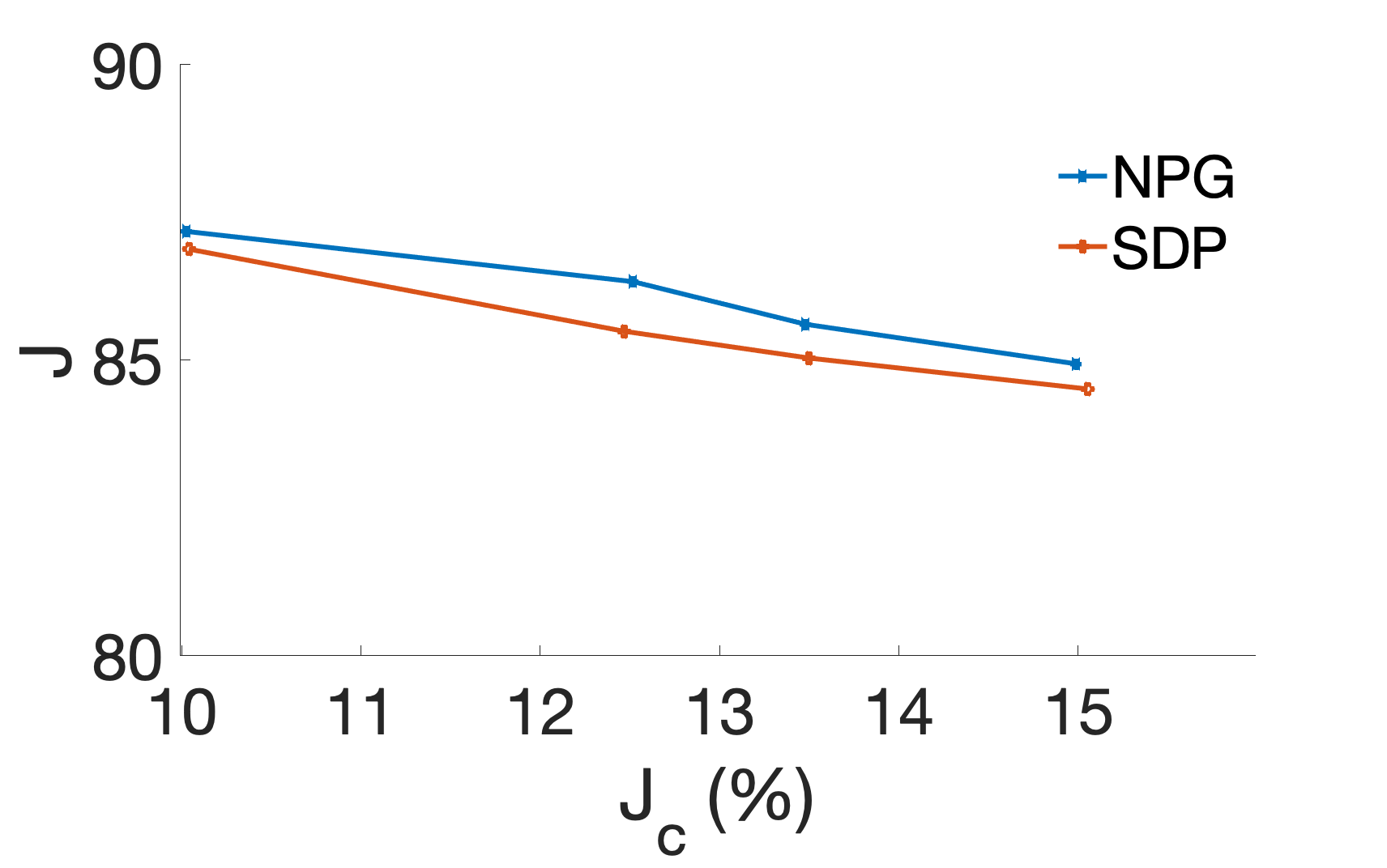}
    \caption{Control cost $J$ vs. Constraint violation probability $J_c$. $\delta = [15.0,13.5, 12.5, 10.0]$.}
    \label{fig:sdp_vs_pd}
\end{figure}



\section{Conclusion} \label{sec:conclusion}
We have considered  PG-based AC algorithms for a probabilistic risk- or chance-constrained LQR problem under the unknown model scenario. The numerical simulations show that the NPG and GNPG-based AC methods exhibit good convergence properties and maintain closed-loop stability while training. On the other hand, DDPG has a larger variance, and the closed-loop system may not remain stable during training.  Finally, we observe that all the PG-based algorithms perform similarly to the chance-contrained LQR and scenario-based MPC technique, which are model-based approaches. Additionally, MPC has certain computational disadvantages compared to the PG-based model-free approaches. Furthermore, we have proved the analytical convergence properties of the NPG and GNPG algorithms under the known model scenario. The proof of convergence for the unknown model scenario is part of our ongoing work.

\appendices

\section{Parameters}
\label{apdx:params} \vspace{-5mm}
\begin{align*}
	&{\bf A} =\begin{bmatrix}
		1 & 0.5 & 0 & 0  \\
		0 & 1 & 0& 0 \\
		0 & 0 & 1 & 0.5 \\
		0 &0 & 0 & 1
	\end{bmatrix}, {\bf B} =\begin{bmatrix}
	0.125 & 0  \\
	0.5 & 0 \\
	0 & 0.125 \\
	0 & 0.5
\end{bmatrix},  \\
	&{\bf W} = diag\left(1, 0.1, 2, 0.2 \right), {\bf U} = {\bf I} ,  \epsilon = 5, \Sigma_w = diag(80,0.01) \\
    &{\bf q} = \left[1, 0.1, 2, 0.2 \right]^T, \alpha_c = 0.005, \alpha_a = 0.005, \alpha_d = 0.001, \\
    & \Sigma_{D,0} = 5{\bf I}, \Sigma_{D,F} = 0.01{\bf I}, \Sigma_u = {\bf I}. 
\end{align*}

\section{Proof of Lemma \ref{lemma:grad_dom}} \label{apdx:grad_dom}
From Lemma~C.6 in \cite{yangProvablyGlobalConvergence2019}, we can write
\begin{align}
    J({\bf K}) - J({{\bf \bar K}}^*) \le \frac{||\Sigma_{\bar K^*}||}{\sigma_{\min}(\bf R)}\text{tr}({\bf E}_K^T{\bf E}_K), \label{eqn:diff_JK}
\end{align}
where ${\bf \bar K}^*$ is the optimal policy parameter that miminizes only the cost function $J(\bf K)$ and 
\begin{align}
    {\bf E}_K &= ({\bf R} + {\bf B}^T{\bf P}_{ K}{\bf B}){\bf K} - {\bf B}^T{\bf P}_{ K}{\bf A}. \label{eqn:E_K}
\end{align}
${\bf K^*}$ is the optimal policy parameter that minimizes the Lagrangian function $\mathcal{L}(\bf K, \lambda)$ for a given $\lambda > 0$. From (\ref{eqn:diff_JK}), we can write
\begin{align}
    &J({\bf K}) - J({{\bf K}}^*) \le J({\bf K}) - J({{\bf \bar K}}^*) \le \frac{||\Sigma_{\bar K^*}||}{\sigma_{\min}(\bf R)}\text{tr}({\bf E}_K^T{\bf E}_K). \nonumber \\
    &\le \frac{||\Sigma_{\bar K^*}||}{4\sigma_{\min}(\bf R)}\text{tr}(4{\bf \Sigma_K}^{-1}{\bf \Sigma_K}{\bf E}_K^T{\bf E}_K{\bf \Sigma_K}^{-1}{\bf \Sigma_K}). \nonumber \\
    &\le \frac{n||\Sigma_{\bar K^*}||}{4\sigma^2_{\min}(\Sigma_K) \sigma_{\min}(\bf R)}\mid \mid \nabla_K {J}({\bf K})  \mid \mid^2 \le \mu_1 \mid \mid \nabla_K {J}({\bf K})  \mid \mid^2,  \nonumber \\
    & \text{where } \mu_1 = \frac{n||\Sigma_{\bar K^*}||}{4\sigma^2_{\min}(\Sigma_K) \sigma_{\min}(\bf R)}. \label{eqn:diff_JK2}
\end{align}
In (\ref{eqn:diff_JK2}), we have used the following results from \cite{yangProvablyGlobalConvergence2019},
\begin{align}
    \nabla_K {J}({\bf K}) = 2 {\bf E}_K{\Sigma_K}. \label{eqn:grad_JK} 
\end{align}
Taking derivative of (\ref{eqn:JcK}) with respect to $\bf K$ we can write,
\begin{equation}
    \nabla_K {J_c}({\bf K}) = -\text{E}\left[\exp\left(-a({\bf x}_k, {\bf K})^2 /2\right) \frac{{\bf B}^T{\bf q}{\bf x}_k^T}{\sqrt{2\pi {\bf q}^T{\Sigma_{\bar w}}{\bf q}}}  \right]. \label{eqn:grad_JcK}
\end{equation}
Norm of the gradient of the Lagrangian function $\mathcal{L}(\bf K, \lambda)$ can be written as, 
\begin{align} 
&\text{tr}\left(\nabla_K\mathcal{L}({\bf K}, \lambda)^T \nabla_K\mathcal{L}({\bf K}, \lambda) \right) = \text{tr}\left(\nabla_K J({\bf K})^T \nabla_K J({\bf K}) \right) \nonumber \\
&+ \text{tr}\left(\lambda^2\nabla_K J_c({\bf K})^T \nabla_K J_c({\bf K}) + 2\lambda\nabla_K J({\bf K})^T\nabla_K J_c({\bf K})  \right). \nonumber \\
& \ge \text{tr}\left(\nabla_K J({\bf K})^T \nabla_K J({\bf K}) \right) + \text{tr}\left(2\lambda\nabla_K J({\bf K})^T\nabla_K J_c({\bf K})  \right) \nonumber \\
& \ge \text{tr}\left(\nabla_K J({\bf K})^T \nabla_K J({\bf K}) - 4\lambda {\bf E}_K{\Sigma_K}\text{E}\left[\frac{{\bf B}^T{\bf q}{\bf x}_k^T}{\sqrt{ 2\pi{\bf q}^T{\Sigma_{\bar w}}{\bf q}}} \right]  \right) \nonumber \\
&\text{[using (\ref{eqn:grad_JK}) and (\ref{eqn:grad_JcK}), and } 0\le \exp\left(-a^2({\bf x}_k, {\bf K}) /2\right) \le 1. \text{]} \nonumber \\
& \ge \text{tr}\left(\nabla_K J({\bf K})^T \nabla_K J({\bf K}) \right)\text{ [since } \text{E}[{\bf x}_k] = 0 \text{]}.
\label{eqn:grad_LK}
\end{align}
Combining (\ref{eqn:diff_JK2}) and (\ref{eqn:grad_LK}), we can write
\begin{align} 
    J({\bf K}) - J({{\bf K}}^*) \le \mu_1 \text{tr}\left(\nabla_K\mathcal{L}({\bf K}, \lambda)^T \nabla_K\mathcal{L}({\bf K}, \lambda) \right).
\end{align}
Additionally, from Lemma C.6 \cite{yangProvablyGlobalConvergence2019}, we can say $\nabla_K\mathcal{L}({\bf K}, \lambda)^T \nabla_K\mathcal{L}({\bf K}, \lambda)$ is lower bounded away from $0$ by $\sigma(\Sigma_w)||{\bf R} + {\bf B}^T{\bf P}_K{\bf B}||^{-1}\text{tr}({\bf E}_K^T{\bf E}_K)$. Since $\mid J_c({\bf K}) - J_c({\bf K^*}) \mid \le 1$, there will exist a sufficiently large $\mu \ge \mu_1$, such that (we note that $\mu$ is dependent on $\lambda$)
\begin{align} 
    &\mathcal{L}({\bf K}, \lambda) - \mathcal{L}({\bf K}^*, \lambda) = J({\bf K}) - J({{\bf K}}^*) + \lambda (J_c({\bf K}) - J_c({\bf K^*})) \nonumber \\
    & \le \mu \text{tr}\left(\nabla_K\mathcal{L}({\bf K}, \lambda)^T \nabla_K\mathcal{L}({\bf K}, \lambda) \right).
\end{align}
This completes the proof of Lemma~\ref{lemma:grad_dom}.

\section{Proof of Lemma~\ref{lemma:conv_rate}} \label{apdx:convergence_rate}
Here, we will prove Lemma~\ref{lemma:conv_rate} for the NPG algorithm. The update rule for the policy parameter ${\bf K}$ under the NPG algorithm is given by (\ref{eqn:K_update}).

From the L-smoothness property of $\mathcal{L}({\bf K},\lambda)$ as given in Lemma~\ref{lemma:Lsmooth}, we can write the following inequality \cite{huTheoreticalFoundationPolicy2023},
\begin{align}
    &\mathcal{L}({\bf K}^{'}, \lambda) - \mathcal{L}({\bf K}, \lambda) \le \nonumber \\
    &\text{tr}(\nabla_K\mathcal{L}({\bf K}, \lambda)^T({\bf K}^{'} - {\bf K})) + \frac{L}{2} \mid \mid {\bf K}^{'} - {\bf K} \mid \mid^2, \label{eqn:grad_dom2}
\end{align}
Using (\ref{eqn:K_update}) in (\ref{eqn:grad_dom2}), we can write
\begin{equation}
\begin{aligned}
    &\mathcal{L}({\bf K}^{'}, \lambda) - \mathcal{L}({\bf K}, \lambda) \le  -\text{tr}\left(\alpha {\Sigma_K}^{-1} \right.\\
    & \left. - \frac{L\alpha^2}{2}({\Sigma_K}^{-1})({\Sigma_K}^{-1})^T \right)\mid \mid \nabla_K \mathcal{L}({\bf K}, \lambda)  \mid \mid^2. \label{eqn:grad_dom3}
\end{aligned}
\end{equation}
We have used the matrix trace inequality as given in Theorem~1 from \cite{coope1994matrix} to get (\ref{eqn:grad_dom3}). To ensure convergence, we need the trace in (\ref{eqn:grad_dom3})) to be strictly positive. In other words, the step size $\alpha$ should be $\alpha < \frac{2}{L\text{tr}(\Sigma_K^{-1})}$.
Since $\Sigma_K$ is the solution to the Lyapunov equation (\ref{eqn:Sigma_K}), we can say $\text{tr}(\Sigma_K)$ is upper bounded by a finite constant, so there exists a constant $0<C_{\Sigma} < \text{tr}(\Sigma_K^{-1})$. Therefore, we can write an upper limit for $\alpha < \frac{2}{LC_{\Sigma}}$, which is independent of $\bf K$.

Applying the gradient dominance property of $\mathcal{L}({\bf K}, \lambda)$ as given in Lemma~\ref{lemma:grad_dom}, we can write 
\begin{align}
    &\mathcal{L}({\bf K}^{'}, \lambda) - \mathcal{L}({\bf K}^*, \lambda) \le \beta (\mathcal{L}({\bf K}, \lambda) - \mathcal{L}({\bf K}^*, \lambda)) \text{, where} \label{eqn:grad_dom4} \\
    &\beta =1- \frac{1}{\mu}\text{tr}\left(\alpha {\Sigma_K}^{-1} - \frac{L\alpha^2}{2}({\Sigma_K}^{-1})({\Sigma_K}^{-1})^T \right) \label{eqn:phi}
\end{align}
Since we need $0<\beta < 1$ for convergence, the step size $\alpha$ should satisfy the following condition
\begin{align}
    &0 < \frac{1}{\mu}\text{tr}\left(\alpha {\Sigma_K}^{-1} - \frac{L\alpha^2}{2}({\Sigma_K}^{-1})({\Sigma_K}^{-1})^T \right) < 1 \nonumber \\
    &=> \frac{1}{\mu}\text{tr}\left(\alpha {\Sigma_K}^{-1} - \frac{L\alpha^2}{2}({\Sigma_K}^{-1})({\Sigma_K}^{-1})^T \right) < 1 \text{ [if $\alpha < \frac{2}{LC_{\Sigma}}$]} \nonumber \\
    &=> \frac{1}{\mu}\text{tr}\left(\frac{\alpha}{\sigma_{\min}(\Sigma_{\bar w})} - \frac{L\alpha^2}{2}C_{\Sigma}^2 \right) < 1 \text{, [(\ref{eqn:Sigma_K}) used]}
    \label{eqn:alpha_bound2}
\end{align}
Here, $\sigma_{\min}(\cdot)$ denotes the lowest singular value. Note that if we choose $\alpha$ sufficiently small, condition (\ref{eqn:alpha_bound2}) can be satisfied. This completes the proof of Lemma~\ref{lemma:conv_rate}.

    \section{Proof of Lemma \ref{lemma:duality}} \label{apdx:duality}
    We follow the proof of Theorem 2 from \cite{zhaoGlobalConvergencePolicy2023a}. The proof contains two steps. 

    First, it is proved that there exists a $\lambda^* \triangleq \inf \left \{ \lambda \ge 0|J_c({\bf K}^*(\lambda)) \le \delta  \right \}$ such that $\lambda^* < \infty$. Although the constraint function differs in our case, we can utilize the same proof methodology as presented in \cite{zhaoGlobalConvergencePolicy2023a}, which relies on a contradiction argument employing Slater's condition. This proof does not rely on any specific formulation of the constraint function.

For the second step of the proof, we need to show that $\bf K^*(\lambda)$ and $J_c(\bf K^*(\lambda))$ are continuous functions of $\lambda$. We will prove this step in the following.  
We can directly say the gradient of the Lagrangian function $\mathcal{L}(\bf K, \lambda)$ with respect to $\bf K$ is a linear function of $\lambda$ for a fixed $\bf K$. Additionally, $\nabla_K \mathcal{L}(\bf K, \lambda)$ is continuous in $K \in \mathcal{K}$, see Lemma~\ref{lemma:twice_diff}. Therefore, the policy gradient steps, see Algorithm~\ref{algo:NPG}, will produce $\bf K_i$ that are continuous functions of $\lambda$. Finally, we have already proved that $\bf K_i \rightarrow \bf K^*$ as $i \rightarrow \infty$ in Lemma~\ref{lemma:grad_dom}. Therefore, we can say the optimal policy parameter $\bf K^*(\lambda)$ and the constraint function $J_c(\bf K^*(\lambda))$ are continuous functions of $\lambda$. This completes the proof of Lemma~\ref{lemma:duality}.

\bibliographystyle{IEEEtran}
\bibliography{IEEEabrv,Const_Control} 

\end{document}